\documentclass[11pt]{article}
\usepackage{authblk}

\usepackage{epsfig}
\usepackage{amsmath}
\usepackage{amsthm}
\theoremstyle{theorem}
\newtheorem{theorem}{Theorem}
\theoremstyle{definition}

\newtheorem{example}{Example}
\usepackage{amssymb}
\usepackage{color}
\usepackage{url}
\usepackage{algorithm}
\usepackage{algorithmic}
\usepackage[ruled,linesnumbered,algo2e,vlined]{algorithm2e}

\usepackage[top=3cm, bottom=3cm, left=3.5cm, right=3.5cm, includefoot]{geometry}

\newcommand{\derive}[2]{\Rightarrow_{#1}^{#2}}
\newcommand{\lrceil}[1]{\lceil #1 \rceil}
\newcommand{\Fib}{\mathit{Fib}}

\begin{document}
\title{Grammar compression with\\ probabilistic context-free grammar}

\author[1]{Hiroaki Naganuma}
\author[1]{Diptarama Hendrian}
\author[1]{Ryo Yoshinaka}
\author[1]{Ayumi Shinohara}
\author[2]{Naoki Kobayashi}

\affil[1]{Tohoku University, Japan}
\affil[2]{The University of Tokyo, Japan}

\date{}
\maketitle 
\begin{abstract}
We propose a new approach for universal lossless text compression, based on grammar compression.
In the literature, 
a target string $T$ has been compressed as a context-free grammar $G$ in Chomsky normal form satisfying $L(G) = \{T\}$.
Such a grammar is often called a \emph{straight-line program} (SLP).
In this paper, we consider a probabilistic grammar $G$ that generates $T$, but not necessarily as a unique element of $L(G)$.
In order to recover the original text $T$ unambiguously, we keep both the grammar $G$ and the derivation tree of $T$ from the start symbol in $G$, in compressed form.
We show some simple evidence that our proposal is indeed more efficient than SLPs for certain texts, both from theoretical and practical points of view.

\end{abstract}

\section{Introduction}
There are two main approaches of \emph{grammar-based compressions} in the literature.
One is a domain-specific approach, where
a grammar $G$ is fixed and known to both an encoder and a decoder, and 
any target text $T$ is assumed to belong to the language $L(G)$. 
When compressing the text $T$, the encoder compresses the derivation tree of $T$ from the start symbol of the grammar $G$.
Katajainen et al.~\cite{Katajainen1986software} and Cameron~\cite{Cameron1988source} initiated 
this approach for compressing program source codes, and
Eck et al.~\cite{Eck1998dcc} and Evans~\cite{Evans1998dcc} followed them.
Lake~\cite{Lake2000dcc} and Tarhio~\cite{Tarhio2001CPT} strengthened their work by combining with the \emph{partial pattern matching (PPM) method}~\cite{Cleary1984ppm}. Depending on the target domains, various methods are proposed, especially for 
XML files~\cite{Cheney2001dcc,liefke2000xmill}.

The other approach is universal lossless compression, which assumes no specific grammar by which target texts should be derived.
Rather, we must find a grammar $G$ that \emph{exactly} captures the structure of the input text $T$ in the sense that $L(G) = \{T\}$.
Since $L(G)$ is a singleton, the original text $T$ can successfully be recovered from $G$, where we do not need a derivation tree any more.
If the description size of the grammar $G$ is much smaller than the size of the original text $T$, 
we can 
not only save the space, but also obtain some hierarchical structures hidden in $T$.
The core idea of this approach can be traced back to the famous compression algorithms~\cite{ZivLempel1978lz,Storer1982jacm,Welch1984computer} in the 70's and 80's.
{\sc Sequitur}~\cite{NevillManning1997CJ} and {\sc Re-Pair}~\cite{Larsson1999RePair} are pioneering work that belong to this approach, and
various improvements have been proposed~\cite{navarro2008re,Yoshida2013dcc,Sekine2014dcc,Masaki2016dcc,Bille2017dcc,Ganczorz2017dcc,furuya2019mr}.
In the literature, a Chomsky-normal form grammar $G$ that generates a singleton $\{ T \}$ is often called a \emph{straight-line program} (SLP)~\cite{Karpinski1995cpm,Maruyama2014fully,Saad2018dcc}, and some extensions to tree grammars~\cite{Lohrey2011dcc,Lohrey2015dlt} are also known.

In this paper, we propose a new compression scheme which generalizes those two approaches.
For each target text $T$, we consider a different grammar $G$ such that $T \in L(G)$, but $L(G)$ is not necessarily a singleton.
To ensure that \(T\) can be uniquely reconstructed by a decoder,
the encoder needs to find a grammar $G$ deriving $T$ and encode both the grammar $G$ and the derivation tree of $T$ in $G$.
Although the proposed scheme itself should widely be applicable for various classes of grammars, 
we focus on the \emph{probabilistic context-free grammars in Chomsky-normal form (PCFGs)}
in the present paper, in order to compare them with SLPs,
both from theoretical and practical points of view.
We provide some simple facts and observations that our scheme is promising to have more compact expressions of certain texts.
In Section~\ref{sec:gc}, we prove that a string ${\tt a}^n$ can be expressed in $O(\log{n})$ bits in PCFGs, while it requires $\Omega(\log{n}\log{\log{n}})$ bits in SLPs.
Experimental results on Fibonacci strings with noises are presented in Section~\ref{sec:experiments}.
Fibonacci strings are among the most efficiently compressible strings by SLPs.
We will see that while noises disturb much the regular structure of the SLPs, augmenting them with auxiliary rules for noises with low probability accomplishes efficient compression of the noisy Fibonacci strings.

\section{Preliminaries}
\label{sec:preliminaries}
A \emph{context-free grammar (CFG)} is denoted as a quadruple $G=(\Sigma,V,R,S)$, where 
$\Sigma$ is the terminal symbol set, 
$V$ is the nonterminal symbol set, 
$R \subseteq V \times (\Sigma \cup V)^*$ is the production rule set,
and $S \in V$ is the start symbol.
The \emph{head} and the \emph{body} of a rule $v \to \alpha \in R$ are $v$ and $\alpha$, respectively.
The set of rules with head $v \in V$ is denoted as $R_v$.
Each rule $r \in R$ defines a binary relation $\derive{r}{}$: for $r = v \to \alpha$, we write $\beta v \gamma \derive{r}{} \beta \alpha \gamma$ for any $\beta,\gamma \in (\Sigma \cup V)^*$.
We also write $\alpha \derive{G}{} \beta$ if there is $r \in R$ for which $\alpha \derive{r}{} \beta$.
The reflexive and transitive closure of $\derive{G}{}$ is denoted by $\derive{G}{*}$.
The language of $G$, denoted by $L(G)$, is the set of strings over $\Sigma$ that are obtained by rewriting strings using production rules starting from $S$, i.e., $L(G) = \{\, T \in \Sigma^* \mid S \derive{G}{*} T\,\}$.
A rewriting $\beta v \gamma \derive{G}{} \beta \alpha \gamma$ is called \emph{left-most} if $\beta \in \Sigma^*$.
A \emph{left-most derivation sequence} of a string $T \in \Sigma^*$ by $G$ is a sequence of rules of $G$ that derives $T$ by left-most rewritings.
Since we consider only left-most derivations in this paper, by a \emph{derivation sequence} we mean a left-most derivation sequence.
Given a grammar $G$ and a derivation sequence of a string $T \in L(G)$, one can recover the string $T$.
A grammar is in Chomsky-normal form (CNF) if $R \subseteq V \times (\Sigma\cup V)^2$ and $S$ occurs in the body of no rules.
A \emph{straight-line grammar (SLG)} is a CFG which generates exactly one string, i.e., $L(G)$ is a singleton.
We have $|R_v|=1$ for an SLG unless it has useless rules.
An SLG in CNF is called a \emph{straight-line program (SLP)}.
A \emph{probabilistic grammar} $G_\pi$ is a pair of a CFG $G=(\Sigma,V,R,S)$ and a probability function $\pi:R \to [0,1]$ such that $\sum_{r \in R_v} \pi(r) = 1$ for each $v \in V$.
The probability on rules is generalized to derivation sequences as the product of rules constituting the sequence, i.e., $\pi(r_1 \dots r_m) = \pi(r_1) \dots \pi(r_m)$.
The probability $\pi(T)$ of a string $T$ is the sum of the probabilities of all the derivation sequences of $T$.\footnote{Precisely speaking, not every PCFG defines a probability distribution over $\Sigma^*$, but we ignore such technical details as it does not matter for our discussions in this paper.}

\begin{example}\label{ex:cfg}
Consider a grammar $G = (\Sigma,V,R,S)$ where $\Sigma = \{{\tt a},{\tt b},{\tt c}, {\tt d}, {\tt z}\}$, $V=\{v_1,\dots,v_9,S\}$ and $R$ consists of
\[
\begin{array}{llll}
	r_{1,0}: v_1 \to {\tt a},	&	r_{2,0}: v_2 \to {\tt b},	&	r_{3,0}: v_3 \to {\tt c},	&	r_{4,0}: v_4 \to {\tt d},
\\	r_{5,0}: v_5 \to {\tt z},	&	r_{6,0}: v_6 \to v_1v_2,	&	r_{6,1}: v_6 \to v_5 v_2,	&	r_{7,0}: v_7 \to v_3 v_4,
\\	r_{8,0}: v_8 \to v_6v_7,	&	r_{9,0}: v_9 \to v_8v_9,	&	r_{9,1}: v_9 \to z,	&	r_{S,0}: S \to v_9.
\end{array}
\]
The language of $G$ is $L(G) = (({\tt a}|{\tt z}){\tt bcd})^*{\tt z}$.
The left-most derivation sequence of ${\tt abcdzbcdz}$ is $\rho = r_{S,0} r_{9,0} r_{8,0} r_{6,0} r_{1,0} r_{2,0}
 r_{7,0} r_{3,0} r_{4,0} r_{9,0}
 r_{8,0} r_{6,1} r_{5,0}
 r_{8,0} r_{7,0} r_{3,0} 
 r_{4,0} r_{9,1}$, for
\begin{align}
    \label{derivation_process}
    S &\derive{r_{S,0}}{} v_9 \derive{r_{9,0}}{} v_8v_9 \derive{r_{8,0}}{} v_6v_7 v_9 \derive{r_{6,0}}{} v_1v_2v_7v_9 \derive{r_{1,0}}{} {\tt a}v_2v_7v_9 \derive{r_{2,0}}{} {\tt ab}v_7v_9 \nonumber \\
    &\derive{r_{7,0}}{} {\tt ab} v_3v_4v_9 \derive{r_{3,0}}{} {\tt abc}v_4v_9 \derive{r_{4,0}}{} {\tt abcd}v_9 \derive{r_{9,0}}{} {\tt abcd}v_8v_9 \nonumber \\
    &\derive{r_{8,0}}{} {\tt abcd} v_6v_7v_9 \derive{r_{6,1}}{} {\tt abcd}v_5v_2v_7v_9 \derive{r_{5,0}}{} {\tt abcdz}v_2v_7v_9 \\
    &\derive{r_{8,0}}{} {\tt abcdzb} v_7v_9 \derive{r_{7,0}}{} {\tt abcdzb}v_3v_4v_9 \derive{r_{3,0}}{} {\tt abcdzbc}v_4v_9 \nonumber \\
    &\derive{r_{4,0}}{} {\tt abcdzbcd} v_9 \derive{r_{9,1}}{} {\tt abcdzbcdz} \,.\nonumber
\end{align}
Define $\pi: R \to [0,1]$ by $\pi(r) = 1$ for all $r \in R$ but $\pi(r_{6,0})=\pi(r_{6,1})=0.5$, $\pi(r_{9,0})=0.7$, and  $\pi(r_{9,1})=0.3$.
Then $\pi(\rho) = 0.7 \cdot 0.5 \cdot 0.7 \cdot 0.5 \cdot 0.3 = 0.03675$.
\end{example}

\section{Our proposal framework of grammar compression using PCFGs}

\label{sec:gc}
Domain-specific grammar compression techniques assume that the texts to be compressed
are derived from a specific CFG, which is known to both an encoder and a decoder.
In this setting, we take advantage of the grammar knowledge and
encode and decode a derivation tree, or equivalently a derivation sequence, for the input text.
If we assume an appropriate probability assignment to the rules of the grammar, 
the derivation sequence may be highly compressed using the standard arithmetic coding technique~\cite{Cameron1988source}.
If the PCFG gives a probability $\pi(\rho)$ to a derivation sequence $\rho$, at most $\lrceil{- \log_2 \pi(\rho)}$ bits will suffice to encode $\rho$.

On the other hand, the universal grammar compression does not assume any specific grammar a priori.
The standard technique compresses an input text $T$
by finding a small SLG $G$, particularly an SLP, such that $L(G)=\{T\}$.
An SLG can be seen as a special case of a PCFG, where every rule has probability 1.
Accordingly, we use no bits to remember the unique derivation sequence, but we must remember the grammar itself.

This paper proposes a combination of those approaches.
We compress a text by both a PCFG and a derivation sequence for the text.
A \emph{PCFG compression} of a text $T$ is an encoded pair of a PCFG $G_\pi$ and a derivation sequence of $T$.
In the domain-specific approach, we encode a derivation sequence but no need to encode the grammar,
while it is the other way around in the standard universal compression.
Our approach might appear less efficient as we have to remember both.
However, we establish the following theorem with a concrete example. 
\begin{theorem}
There is an infinite family of texts which any SLPs require size $\Omega(\log n \log \log n)$ to derive,
while PCFGs with arithmetic coding of a derivation sequence have size $\Theta(\log n)$.
\end{theorem}
\begin{proof}
Let $T = {\tt a}^{n}$ with $n=2^m$ for some $m \ge 0$.
This is one of the strings that the standard grammar compression techniques compress most effectively.
The string is derived by the SLP consisting of rules of the form $v_{i+1} \to v_{i} v_i$ for $0 \le i < m$ and $v_0 \to {\tt a}$ with $v_m$ being the start symbol.
This SLP has only $m+1$ rules.
Since we require $\lrceil{\log (m+1)}$ bits to distinguish nonterminals, the description size will be $\Theta(m \log m) = \Theta(\log n \log\log n)$.

Consider the CFG $G$ with just two rules $r_0:S \to {\tt a} S$ and $r_1:S \to {\tt a}$,
which derives $T= {\tt a}^{n}$ with the derivation sequence $r_0^{n-1} r_1$.
We assign probabilities to those rules as $\pi(r_0) = 1-2^{-m}$ and $\pi(r_1) = 2^{-m}$ in accordance with their occurrence frequencies in the derivation sequence.
The size of $G$ is constant.
There can be different ideas to encode $\pi$, but anyway it would require at most $m$ bits. 
The derivation sequence $r_0^{n-1} r_1$ is represented by $- \log_2 \pi(r_0^{n-1} r_1) = - \log_2 (1-1/n)^{n-1}2^{-m} = m - \log_2 (1-1/n)^{n-1} = \Theta(m)$ bits using the standard arithmetic coding.
All in all, this compression with the PCFG requires $\Theta(\log n)$ bits. 
\end{proof}

In addition to the theoretical analysis on the benefit of our compression framework with PCFGs,
we will show experimental results in the remainder of this paper.

\section{Experiments}
\label{sec:experiments}
\subsection{Data used in our experiments}
Throughout our experiments, we use Fibonacci strings with noises.
The $m$-th Fibonacci string $\Fib_m$ $(m \geq 0)$ is defined by $\Fib_0={\tt b}$, $\Fib_1= {\tt a}$, and $\Fib_m = \Fib_{m-1}\cdot\Fib_{m-2}$ for $m \ge 2$.
For instance, $\Fib_{4} = {\tt abaab}$ and $\Fib_{5} = {\tt abaababa}$.
The \emph{Fibonacci SLP} for $\Fib_{m}$ has the following rules in accordance with the definition of the Fibonacci strings:
\[
	 r_{0,0}: v_0 \to \mathtt{b},\quad r_{1,0}: v_1 \to \mathtt{a},\quad r_{i,0}: v_{i} \to v_{i-1} v_{i-2} \text{ for } 2 \le i \le m
\]
where $v_m$ is the start symbol. Obviously the size of the SLP is linear in $m$ and each $v_i$ derives the $i$-th Fibonacci string.
We generate strings by adding one of the following types of noises at random positions in a Fibonacci string:
\begin{description}
    \item[Type~0] Replacing {\tt{a}} by {\tt{b}} or the other way around.
    \item[Type~{\textit{k}}] Replacing a letter by one of the $k$ new letters ${\tt{c}_1}, {\tt{c}_2}, \cdots, {\tt{c}_k}$, where $k \ge 1$.
\end{description}

\subsection{Methods to compare}
We compare \textsc{Re-Pair}, \textsc{gzip}, \textsc{bzip2} and our PCFG compression techniques over noisy Fibonacci strings.
An ideal (at least for small amount of noises) compression PCFG for noisy Fibonacci strings is obtained by modifying the Fibonacci SLP adding some rules corresponding to the noises.
Let $G_0$ and $G_k$ for $k \ge 1$ be grammars having rules of the Fibonacci SLP plus
\begin{itemize}
	\item[($G_0$)] $r_{0,1}: v_0 \to \mathtt{a}$ \ and \ $r_{1,1}: v_1 \to \mathtt{b}$,
	\item[($G_k$)] $r_{0,1}: v_0 \to \mathtt{c}_1$, \ $r_{0,2}: v_0 \to \mathtt{c}_2$, $\dots$, $r_{0,k}: v_0 \to \mathtt{c}_{k}$, and
	\\\  $r_{1,1}: v_1 \to \mathtt{c}_1$, \ $r_{1,2}: v_1 \to \mathtt{c}_2$, $\dots$, $r_{1,k}: v_1 \to \mathtt{c}_k$,
\end{itemize}
respectively.
Note that the number of those additional rules is independent of $m$ of $\Fib_m$.
Those grammars are conveniently defined by taking advantage of our \emph{a priori} knowledge about the generation of the target strings.
They are still useful for the goal of our experiments,
which is to demonstrate the potential of the PCFG compression framework rather than to propose a concrete encoder.
To encode the derivation sequence of a noisy Fibonacci string, it suffices to remember its subsequence consisting of $r_{i,j}$ for $i=0$ or $i=1$, i.e., the rules whose head is either $v_0$ or $v_1$.
Moreover, even replacing both $r_{0,j}$ and $r_{1,j}$ with the same symbol $j$ is allowed for unique decoding.
We use the output of \textsc{RangeCoder}~\cite{Martin1979RangeCoder} for such a string as our encoding of the derivation sequence.
Probability assignment to grammar rules is implicitly done in \textsc{RangeCoder}, where the rules $r_{0,j}$ and $r_{1,j}$ have the same probability for respective $j$.

We also design a prototype PCFG encoder modifying \textsc{Re-Pair}~\cite{Larsson1999RePair}, as shown in Algorithm~\ref{alg:proposal}.
\textsc{Re-Pair} first replaces each terminal symbol $a$ in the text by a unique nonterminal $v_a$ and adding a rule $v_a \to a$.
Hereafter by $V_\Sigma = \{\,v_a \mid a \in \Sigma\,\}$ we denote the set of nonterminals that derive a terminal symbol.
Then \textsc{Re-Pair} repeatedly replaces the most frequently occurring bigrams $xy$ in a text by a new nonterminal symbol $z$ and creates a new rule $z \to xy$, until no bigram occurs more than once.
Precisely speaking, what \textsc{Re-Pair} constructs does not follow the definition of an SLP given in Section~\ref{sec:preliminaries}.
Instead of having a start symbol, the constructed grammar has a \emph{start sequence} $S$ of nonterminals.
Algorithm~\ref{alg:proposal} also does the same initial replacement of terminals $a$ by nonterminals $v_a \in V_\Sigma$ and then repeatedly replaces the most frequent bigram by a new nonterminal $z$, where, in addition, the nonterminal $z$ also replaces another bigram that meets some criterion.
This is determined by the function $\mathit{FindMinBigram}(T,v,xy)$ explained below.
In this way, we construct two rules sharing the same nonterminal as their head, where we call the former one \emph{major} and the latter \emph{minor}.
The functions used in Algorithm~\ref{alg:proposal} are defined as follows.
\begin{itemize}
    \item $\mathit{FindMaxBigram}(T)$ returns the most frequent bigram in $T$ as long as it occurs more than once. If no bigram occurs more than once, it returns $\mathit{None}$.
    \item $\mathit{FindMaxContext}(T,xy)$ returns the symbol $c$ that occurs most frequently immediately before $xy$ in $T$.
    \item $\mathit{FindMinBigram}(T,c,xy)$ returns the bigram $x'y'$ that occurs least frequently (but at least once) immediately after $c$ in $T$ and meets the following condition: either $x'=x$ and $y' \in V_\Sigma$ or $y'=y$ and $x' \in V_\Sigma$.
     If there is no such bigram, it returns $\mathit{None}$.
    \item $\mathit{Replace}(T,x,v)$ replaces all occurrences of $x$ in $T$ with a nonterminal symbol $v$.\footnote{More precisely, this replacement shall be done one by one from left to right, e.g., $\mathit{Replace}(aaaaa,aa,v) = vva$.}
\end{itemize}
\begin{algorithm2e}
	\caption{Construct a CFG ${G}_T$ from input $T \in {\Sigma}^*$}
	\label{alg:proposal}
	\SetKwBlock{Loop}{Loop}{end}
		${\Sigma}_T,V_T,R_T := \emptyset$\;
		\For{each terminal symbol $a$ occurring in $T$}{
				${\Sigma}_T := \{a\} \cup {\Sigma}_T$\;
				$V_\Sigma := \{v_a\} \cup V_\Sigma$\;
				$R_T := \{v_a \rightarrow a\} \cup R_T$\;
				$\mathit{Replace}(T,a,v_a)$\; \label{alg:max_replace}
%			}
		}
        $V_T := V_\Sigma$\;
		\While{$\mathit{FindMaxBigram}(T) \neq \mathit{None}$}{ \label{alg:step_start}
            $b := \mathit{FindMaxBigram}(T)$\;
			$V_T := \{v_b\} \cup V_T$ for a fresh nonterminal $v_b$\;
			$R_T := \{v_b \rightarrow b\} \cup R_T$\; \label{alg:major_rule}
			$\mathit{Replace}(T,b,v_b)$\; \label{alg:min_replace}
			$c := \mathit{FindMaxContext}(b)$\; \label{alg:diff_start}
			$b' := \mathit{FindMinBigram}(T,c,b)$\;
			\If{$b' \neq \mathit{None}$}{
				$R_T := \{v_b \rightarrow b'\} \cup R_T$\; \label{alg:minor_rule}
				$\mathit{Replace}(T,b',v_b)$\;
			} \label{alg:diff_end}
		} \label{alg:step_end}
		$S_T := T$\;
		${G}_T := ({\Sigma}_T,V_T,R_T,S_T)$\;
		\textbf{return} ${G}_T$\;
\end{algorithm2e}
The difference between \textsc{Re-Pair} and our proposal method appears in Lines~\ref{alg:diff_start} to~\ref{alg:diff_end}.
The rules added on Line~\ref{alg:minor_rule} are called minor.
Rules added on Line~\ref{alg:major_rule} will be major if there is a minor rule sharing the same head.
The encoding method for derivation sequences of a noisy Fibonacci string under the obtained grammar is just the same as the one under $G_0$ and $G_k$.
We take the subsequence consisting of major and minor rules, where all major ones are represented by $0$ and the minor ones become $1$.

\subsection{Results}
\begin{figure}[tb]
    \centering
	\includegraphics[width=0.9\textwidth]{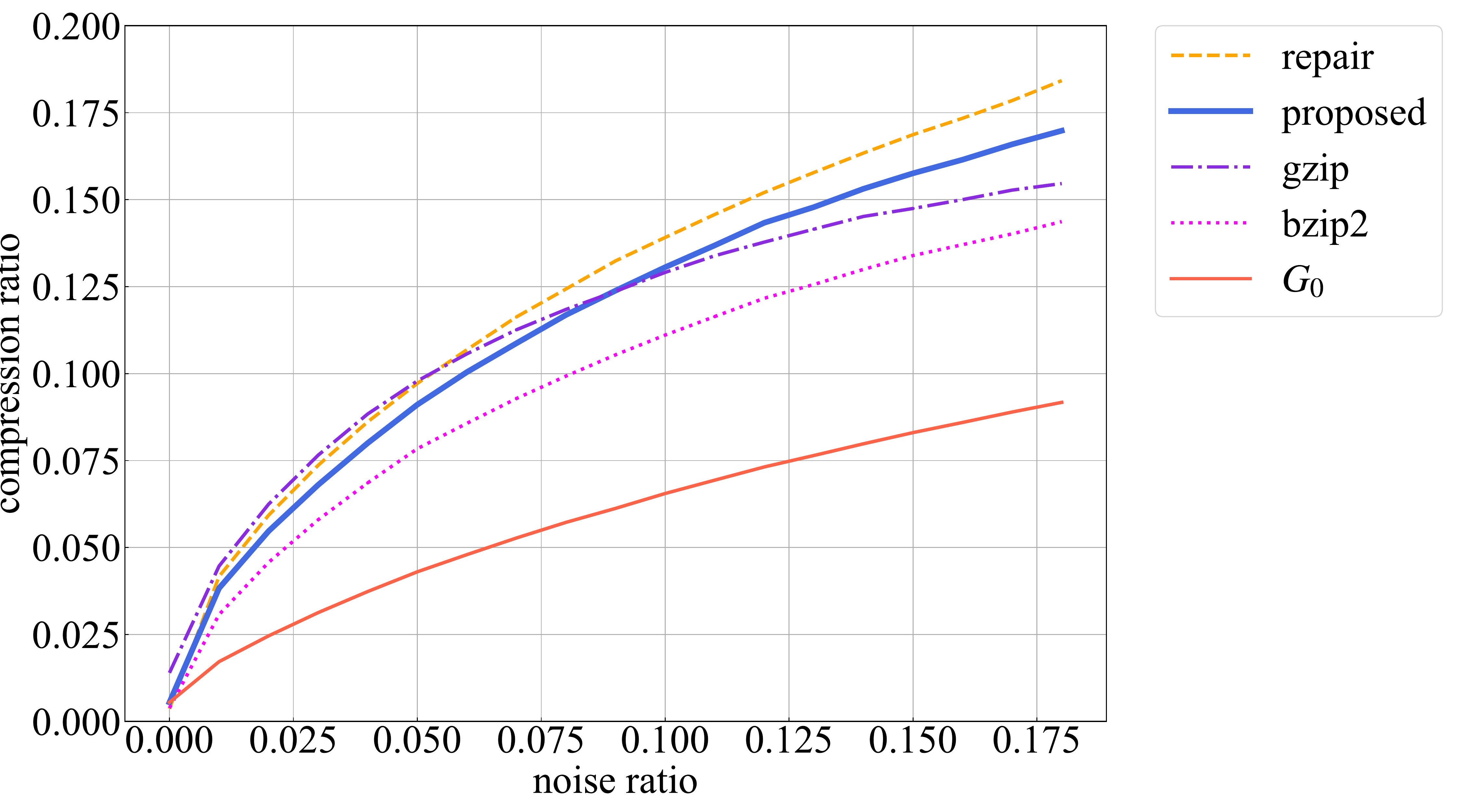}
    \vspace*{-0.5cm}\caption{Compression ratio comparison over Fibonacci strings with Type~0 noise\label{fig:noise1_result}}
\end{figure}

We performed three sets of experiments.
The first and second experiment sets used strings obtained from the 20th Fibonacci string $\Fib_{20}$ by adding noises of Types~0 and~1, respectively, with ratio varying from 0.0\% to 20.0\%, and the third used noises of Type~$k$ for $1 \le k \le 24$ with a fixed ratio $0.1\%$.
The string size is always 10946 bytes.
Figures~\ref{fig:noise1_result},~\ref{fig:noise2_result} and~\ref{fig:noise3_result} respectively compare the compression ratios of different methods against those three sets of noisy Fibonacci strings.
We measured the compression ratios achieved by different methods 10 times for each parameter and used the average.
The compression ratio is defined as the compressed data size over the original text size.
We use Maruyama's implementation for \textsc{Re-Pair}\footnote{https://code.google.com/archive/p/re-pair/},
 Seward's implementation (version~1.0.6) for \textsc{bzip2}, and
Apple gzip~272.250.1 for \textsc{gzip}.

\begin{figure}[tbhp]
    \centering
	\includegraphics[width=0.9\textwidth]{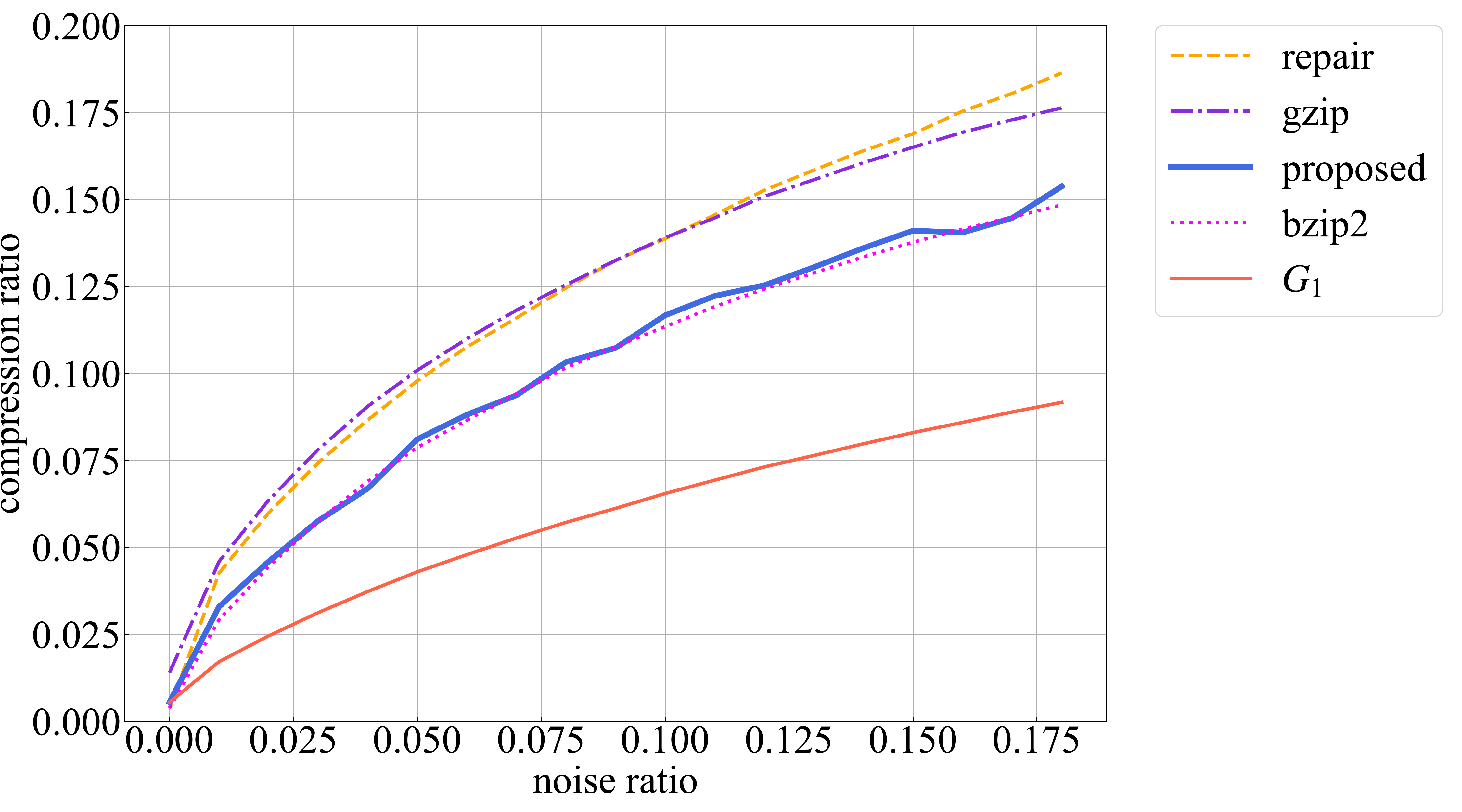}
    \vspace*{-0.5cm}\caption{Compression ratio comparison over Fibonacci strings with Type~1 noise\label{fig:noise2_result}}    
    \vspace*{2\floatsep}
    \centering
    \includegraphics[width=0.9\textwidth]{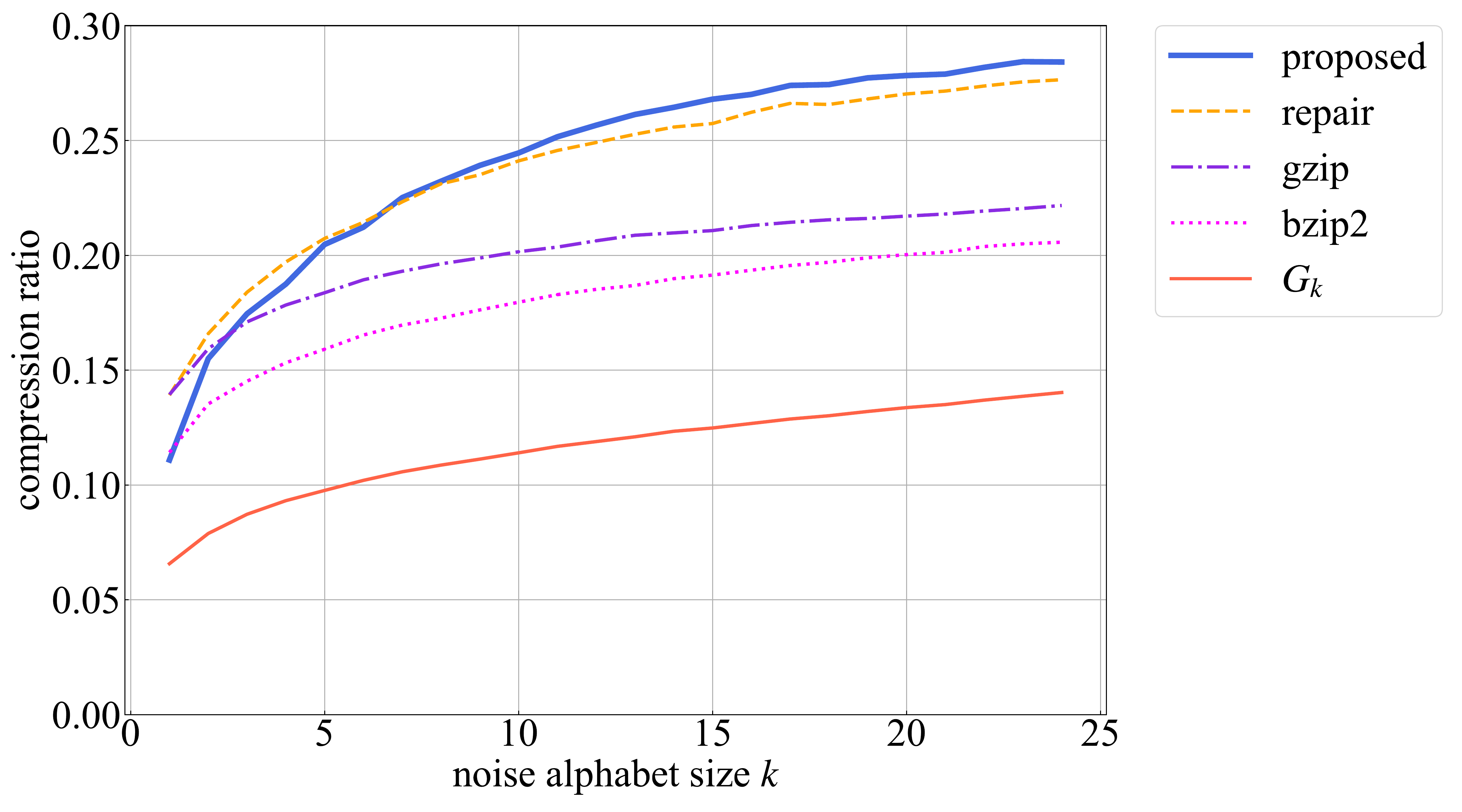}
    \vspace*{-0.5cm}\caption{Compression ratio comparison over Fibonacci strings with Type~$k$ noise\label{fig:noise3_result}}
\end{figure}

Figures~\ref{fig:noise1_result},~\ref{fig:noise2_result} and~\ref{fig:noise3_result} clearly show $G_0$, $G_1$ and $G_k$ achieve the best performance in the respective cases, respectively.  
Algorithm~\ref{alg:proposal} outperformed \textsc{Re-Pair} in the experiments on Type~0 and~$k$ noisy Fibonacci strings for small $k$, while those are almost tied for Type~$k$ noisy strings for bigger $k$.
Figure~\ref{fig:noise2_result} shows that, on Type~1 noisy Fibonacci strings, Algorithm~\ref{alg:proposal} outperformed \textsc{Re-Pair} and \textsc{gzip}, and performed as effectively as \textsc{bzip2}.
When processing Type~1 noisy Fibonacci strings,
the function $\mathit{FindMinBigram}$ tends to return a bigram including $v_{{\tt c}_1}$ at early stages of the while loop.
Then it likely happens that those bigrams including $v_{{\tt c}_1}$, say $v_{{\tt c}_1} v_{\tt b}$, are ``disguised'' as the noise-free bigram, say $v_{\tt a} v_{\tt b}$, of the corresponding major rule by sharing the same head. If indeed those occurrences of ${\tt c}_1$ were originally $\tt b$, this is the ideal behavior.
On the other hand, in Type~$k$ noisy Fibonacci strings for big $k$,
a bigram ${\tt a b}$ may be altered to ${\tt c}_i {\tt b}$ for any $i \le k$, but the nonterminal $v_{\tt ab}$ with the major rule $v_{\tt ab} \to v_{\tt a} v_{\tt b}$ can have only one minor rule, say $v_{\tt ab} \to v_{{\tt c}_1} v_{\tt b}$, where many other bigram $v_{{\tt c}_i} v_{\tt b}$ with $i > 1$ originating in the same bigram ${\tt ab}$ are not treated as its noisy variants.
When processing Type~0 noisy Fibonacci strings, we have no clear mark of noises.
The bigrams including altered letters may appear in the original noise-free Fibonacci string and thus they are not necessarily returned by $\mathit{FindMinBigram}$.
This would explain why Algorithm~\ref{alg:proposal} worked well only on Type~$k$ noisy Fibonacci strings with very small $k$.

\section{Conclusion}
We proposed a new approach to universal lossless text compression
using probabilistic context-free grammars.
We have given some theoretical evidence for the effectiveness of the proposed approach
and confirmed it also through preliminary experiments.
Our proposal framework enables us to represent a noisy text as a pair of the SLP for the ideal noise-free text plus noise information, which can be more compact than the SLP obtained by the standard grammar compression technique.
Various research directions are open for future.
For instance, compressing a collection of similar texts simultaneously with a single grammar and multiple derivations trees would be suitable for our compression scheme.
\emph{Compressed pattern matching}~\cite{Takeda2008cpm} in which texts (and patterns also for some cases) are given as compressed forms, is another interesting problem, because SLPs are often used 
for that purpose~\cite{Karpinski1995cpm,Lifshits2007cpm,Lohrey2012slp,Jez2015recompression}

\bibliographystyle{plain}
\bibliography{refs}

\end{document}